\documentclass[conference]{IEEEtran}

\usepackage{amsmath,amsfonts,amssymb}
\usepackage{hyperref} 
\hypersetup{
	hidelinks,
	colorlinks=true,
	linkcolor=black,
	citecolor=black
}
\usepackage{multirow}
\usepackage{makecell}
\usepackage{cite}

\usepackage{algorithmicx,algorithm}
\usepackage{algpseudocode}

\usepackage{graphicx}
\usepackage{subfigure}
\usepackage{array}
\usepackage{textcomp}
\usepackage{stfloats}
\usepackage{url}
\usepackage{verbatim}
\usepackage{mathrsfs}
\usepackage{enumitem}
\usepackage{bm}
\usepackage{balance}
\usepackage{caption}
\usepackage{xcolor}
\usepackage{amsthm}	

\newtheorem{lemma}{Lemma}

\newtheorem{proposition}{\bf Proposition}

\DeclareMathOperator{\rank}{rank}

\begin{document}

\IEEEoverridecommandlockouts

\title{Channel Estimation for Optical IRS-Assisted \\VLC System via Spatial Coherence}

\author{\IEEEauthorblockN{Shiyuan Sun\textsuperscript{1},Fang Yang\textsuperscript{1}, Weidong Mei\textsuperscript{2}, Jian Song\textsuperscript{1,3}, Zhu Han\textsuperscript{4}, and Rui Zhang\textsuperscript{5,6}}\\ 

\IEEEauthorblockA{
	\textsuperscript{1}Department of Electronic Engineering, Tsinghua University\\
	Beijing National Research Center for Information Science and Technology (BNRist), Beijing 100084, P. R. China\\ 
	\textsuperscript{2}University of Electronic Science and Technology of China, Chengdu 611731, P. R. China\\ 
	\textsuperscript{3}Shenzhen International Graduate School, Tsinghua University, Shenzhen 518055, P. R. China\\ 
	\textsuperscript{4}Department of Electrical and Computer Engineering, University of Houston, Houston, TX 77004 USA\\ 
	\textsuperscript{5}The Chinese University of Hong Kong (Shenzhen), and Shenzhen Research Institute of Big Data, Shenzhen 518172, China\\ 
	\textsuperscript{6}Department of Electrical and Computer Engineering, National University of Singapore, Singapore 117583, Singapore\\
	Email: sunsy20@mails.tsinghua.edu.cn, fangyang@tsinghua.edu.cn, wmei@uestc.edu.cn,\\ jsong@tsinghua.edu.cn, hanzhu22@gmail.com, rzhang@cuhk.edu.cn}}
\maketitle

\begin{abstract}
Optical intelligent reflecting surface (OIRS) has been considered a promising technology for visible light communication (VLC) by constructing visual line-of-sight propagation paths to address the signal blockage issue. 
However, the existing works on OIRSs are mostly based on perfect channel state information (CSI), whose acquisition appears to be challenging due to the passive nature of the OIRS.
To tackle this challenge, this paper proposes a customized channel estimation algorithm for OIRSs. 
Specifically, we first unveil the OIRS spatial coherence characteristics and derive the coherence distance in closed form.
Based on this property, a spatial sampling-based algorithm is proposed to estimate the OIRS-reflected channel, by dividing the OIRS into multiple subarrays based on the coherence distance and sequentially estimating their associated CSI, followed by an interpolation to retrieve the full CSI.
Simulation results validate the derived OIRS spatial coherence and demonstrate the efficacy of the proposed OIRS channel estimation algorithm.
\end{abstract}

\section{Introduction}
To tackle the scarcity of communication spectrum, the visible light communication (VLC) technology has emerged as a pivotal solution, leveraging a 400 THz frequency-band and excelling in energy efficiency~\cite{chi2020visible}. 
Nonetheless, VLC encounters a substantial hurdle in terms of signal blockage, primarily due to the nanoscale wavelength of the lightwave~\cite{feng2018mmwave}.

Fortunately, the emerging intelligent reflecting surface (IRS) offers a viable solution to resolve the blockage issue~\cite{gongshimin2020}. 
Specifically, an IRS consists of numerous reflective elements that can adjust the phase shift and/or amplitude of the incident wave, thereby enabling signal coverage for dead zones~\cite{wu2019towards}. 
The use of IRS for the optical frequency band, namely optical IRS (OIRS), has been explored for various purposes, e.g., enhancing the communication rate performance~\cite{aboagye2022ris,10183987}. 
Particularly, the OIRS-associated channel state information (CSI) needs to be acquired to design the OIRS reflection~\cite{zheng2022survey}.
However, all of the above works assume known CSI, without addressing the channel estimation issue for OIRSs. 

Notably, a variety of channel estimation methods have been developed for IRSs in the radio frequency (RF), which can be broadly classified into two types, i.e., cascaded channel estimation~\cite{zheng2022survey} and separate channel estimation~\cite{hu2021semi}. 
However, these methods cannot be directly applied to OIRSs due to their differences in channel model and signal constraints~\cite{obeed2019optimizing}. 
So far, there are two widely adopted OIRS-reflected channel model, namely, the classical optics model and the association-based model.
In the first model, the OIRS channel gain is directly characterized with respect to (w.r.t.) actual parameters of each reflecting element, e.g., the rotation angles in the mirror-array-based implementation.
Although it offers a precise characterization of the OIRS-reflected channel gain, the theoretical analysis for the optics model can be complicated due to its integral and non-linear expressions~\cite{abdelhady2020visible}.
On the other hand, the OIRS association model characterizes the channel based on the associations between OIRS reflecting elements and the transmitter/receiver antennas~\cite{sun_CL}, under the small-source assumption.
By this means, the OIRS-reflected channel can be expressed in a linear form, which significantly facilitates its theoretical analysis and the channel estimation design.
Nonetheless, to the best of our knowledge, the OIRS channel estimation in VLC system remains largely unexplored under either of the above two channel models.

This paper fills in this gap by proposing an OIRS channel estimation method based on the OIRS association model. 
Specifically, given the widely adopted intensity modulation (IM) in VLC, we unveil that the OIRS-reflected channel shows a prominent coherence in the spatial domain, and proceed to derive the OIRS coherence distance in terms of geometry in closed form. 
Based on this distance, we propose a spatial sampling-based OIRS channel estimation algorithm, which divides the OIRS into multiple subarrays to estimate their associated CSI and recovers the full CSI through spatial interpolation.
Finally, numerical results are provided to validate the effectiveness of the proposed algorithm, which achieves a substantial reduction in estimation overhead with negligible performance loss.

The remainder of the paper is organized as follows. 
Section~\ref{Sec:Model} presents the system model of the OIRS-assisted VLC.
Section~\ref{Sec:Coherence} analyzes the OIRS spatial coherence characteristics.
In Section~\ref{Sec:Proposed}, a spatial sampling-based channel estimation algorithm is proposed.
Numerical results are presented in Section~\ref{Sec:numerical}, and finally the conclusions are drawn in Section~\ref{Sec:Conclude}.

\textit{Notation:}
An uppercase boldface letter $\textbf{A}$ denotes the coordinate in a three-dimensional Cartesian coordinate system, $\textbf{AB}$ denotes a vector from $\textbf{A}$ to $\textbf{B}$, and $\widehat{\textbf{AB}}$ denotes the normalized vector with a magnitude of $\|\widehat{\textbf{AB}}\|_2=1$.
Let $\text{vec}(\cdot)$ denote the vectorization of a matrix by column and $\boldsymbol{I}_N$ denote an $N \times N$ identity matrix.
Moreover, $\otimes$, $\odot$, and $\propto$ represent the Kronecker product, the Hadamard product, and the proportion symbol, respectively.

\section{System Model}
{\label{Sec:Model}}
As depicted in Fig.~\ref{Fig:multiple}, this paper considers a VLC system assisted by an OIRS with $N$ elements, where a receiver equipped with $N_r$ photodetectors (PDs) receives the signal from a transmitter equipped with $N_t$ light-emitting diodes (LEDs).

Compared with conventional RF IRS, where the impinging electromagnetic wave is uniformly reflected in all directions due to the sub-wavelength reflecting elements~\cite{ozdogan2019intelligent}, the OIRS-reflected channel shows different properties.
Specifically, the aperture of each reflecting element far exceeds the wavelength of the lightwave, which leads to an anisotropic property of OIRS, i.e., each OIRS element can be configured to align with a specific pair of LED and PD without causing interference to any other pair, which is known as the ``angle selective'' property~\cite{10190313}.
Accordingly, the OIRS reflection can be characterized based on the alignment or association between its reflecting elements and LED/PD. 
Let $\boldsymbol{G} \triangleq [ \boldsymbol{g}_1, \cdots, \boldsymbol{g}_{N_t} ] \in\{0,1\}^{N\times N_t}$ and $\boldsymbol{F} \triangleq [ \boldsymbol{f}_1, \cdots, \boldsymbol{f}_{N_r} ] \in\{0,1\}^{N\times N_r}$ denote the associations between the OIRS reflecting elements and LEDs and that between OIRS elements and PDs, respectively. 
If an entry of $\boldsymbol{G}$ or $\boldsymbol{F}$ is equal to one, then it implies that the corresponding LED-OIRS element or PD-OIRS element pairs are aligned with each other.
Otherwise, they are not aligned. 
Then, the OIRS-reflected VLC channel can be expressed as~\cite{10190313}
\begin{align}
	\label{Eq:channel}
	\boldsymbol{H} &= \left[
	\begin{matrix}
		\left(\boldsymbol{f}_1 \odot \boldsymbol{g}_1\right)^T\boldsymbol{h}_{1,1}  &  \cdots  &  \left(\boldsymbol{f}_1 \odot \boldsymbol{g}_{N_t}\right)^T\boldsymbol{h}_{1,N_t} \\
		\vdots  &  \ddots  &  \vdots \\
		\left(\boldsymbol{f}_{N_r} \odot \boldsymbol{g}_1\right)^T\boldsymbol{h}_{N_r,1}  &  \cdots  &  \left(\boldsymbol{f}_{N_r} \odot \boldsymbol{g}_{N_t}\right)^T\boldsymbol{h}_{N_r,N_t} \\
	\end{matrix}
	\right],
\end{align}
where $\boldsymbol{h}_{n_r,n_t} \in \mathcal{R}^{N\times 1}_+$ denotes the channel gain vector between the $n_t$th LED and  $n_r$th PD.
The overall CSI parameter matrix to be estimated is defined as $\boldsymbol{H}_c \triangleq [\boldsymbol{h}_{1,1}, \cdots, \boldsymbol{h}_{N_r,1}, \boldsymbol{h}_{1,2}, \cdots, \boldsymbol{h}_{N_r,N_t}]$, which consists of $NN_tN_r$ entries.
Moreover, considering the ``angle selective'' property, the following inequalities should hold, i.e., $\sum_{n_t=1}^{N_t}g_{n, n_t} \leq 1$ and $\sum_{n_r=1}^{N_r}f_{n, n_r} \leq 1$~\cite{10190313}.

We divide the OIRS channel coherence time into multiple time blocks, over which the subarrays of the OIRS are sequentially turned on (to be specified in Section~\ref{Sec:Proposed}). 
Each time block is further divided into $P$ time slots for pilot transmission to estimate the CSI associated with the ``ON'' subarray. 
In the $p$th time slot of any time block, the received signal at the PD is given by
\begin{equation}
	\label{Eq:signal model}
	\setlength\abovedisplayskip{3pt}
	\boldsymbol{y}_p = \boldsymbol{H}\boldsymbol{x}_p + \boldsymbol{z}_p,
	\setlength\belowdisplayskip{3pt}
\end{equation}
where $\boldsymbol{x}_p \in \mathcal{R}^{N_t\times 1}_+$, $\boldsymbol{y}_p \in \mathcal{R}^{N_r\times 1}_+$, and $\boldsymbol{z}_p \in \mathcal{R}^{N_r\times 1}$ represent the transmitted pilot signal, the received signal, and the additive white Gaussian noise (AWGN) in the $p$-th time slot, respectively.
By stacking the received signals over the $P$ slots, namely $\boldsymbol{Y} \triangleq [\boldsymbol{y}_1, \cdots, \boldsymbol{y}_P]$, $\boldsymbol{X} \triangleq [\boldsymbol{x}_1, \cdots, \boldsymbol{x}_P]$, and $\boldsymbol{Z} \triangleq [\boldsymbol{z}_1, \cdots, \boldsymbol{z}_P]$, the overall received signal for a time block can be represented as
\begin{equation}
	\label{Eq:signal model_cohere}
	\setlength\abovedisplayskip{3pt}
	\boldsymbol{Y} = \boldsymbol{H} \boldsymbol{X} + \boldsymbol{Z},
	\setlength\belowdisplayskip{3pt}
\end{equation}
based on which $\boldsymbol{H}_c$ is estimated. 
Note that as compared to the RF IRS, the OIRS channel estimation results in considerably more CSI parameters, i.e., $NN_tN_r \gg N(N_t + N_r)$.
As such, this paper aims to propose a dedicated channel estimation algorithm for OIRSs. 
To this end, we first unveil an inherent coherence characteristic of OIRSs, as detailed below.

\begin{figure}[t]
	\centering    
	\includegraphics[width=0.48\textwidth]{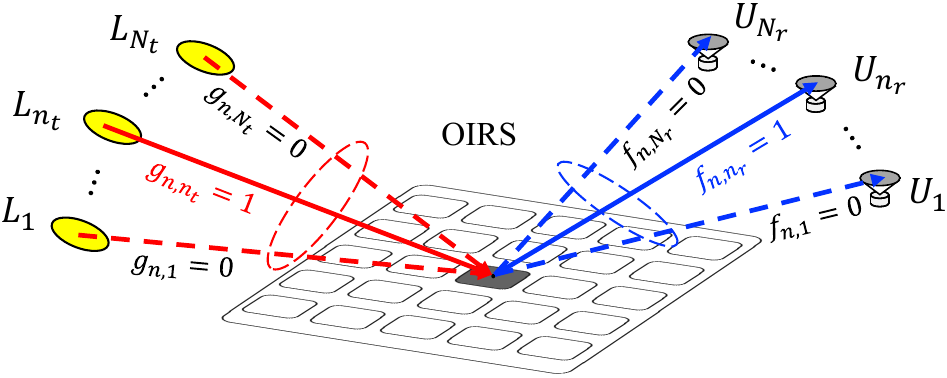}
	\caption{OIRS-reflected channel model.}
	\label{Fig:multiple}
\end{figure}

\section{OIRS Spatial Coherence Analysis}
\label{Sec:Coherence}
In this section, we analyze the spatial coherence of the OIRS channel and derive the OIRS coherence distance.
Considering the coherent modulation in RF communications, as shown in Fig.~\ref{Fig:correlation_RF}, the channel gains between adjacent RF IRS elements can be completely different due to the phase shift.

However, in VLC, the widely adopted IM scheme enables the OIRS-reflected channel gain to change slowly w.r.t. the space shift, as depicted in Fig.~\ref{Fig:correlation_OIRS}.
Let $\textbf{L}$, $\textbf{R}$, and $\textbf{U}$ represent the points at the LED, an OIRS reflecting element, and the PD, respectively.
Under the point source assumption, the OIRS-reflected channel gain in indoor VLC systems is upper-bounded by~\cite{abdelhady2020visible}
\begin{align}
	\label{Eq:lambertian}
	\setlength\abovedisplayskip{3pt}
	h(\textbf{R}) \propto \left\{
	\begin{aligned}
		&\frac{\left(\widehat{\boldsymbol{N}}_1^T\widehat{\textbf{LR}}\right)^m\widehat{\boldsymbol{N}}_2^T\widehat{\textbf{UR}}}{\left(\|\textbf{LR}\|_2 + \|\textbf{UR}\|_2\right)^2}, & &\text{if}\ 0 \leq \widehat{\boldsymbol{N}}_2^T\widehat{\textbf{UR}} \leq \cos\Phi_0,   \\
		&0, & &\text{otherwise},
	\end{aligned}
	\right.
	\setlength\belowdisplayskip{3pt}
\end{align}
\vspace{-0.1cm}

\noindent
where $\widehat{\boldsymbol{N}}_1$, $\widehat{\boldsymbol{N}}_2$, and $\Phi_0$ denote the normal vector at the LED, the normal vector at the PD, and the semi-angle of field-of-view, respectively.
Based on~(\ref{Eq:lambertian}), the relative growth rate of the OIRS-reflected channel gain w.r.t. the space shift $\Delta\textbf{R}$ at $\textbf{R}$ is defined as
\begin{align}
	\label{Eq:xis}
	\setlength\abovedisplayskip{3pt}
	\xi(\Delta\textbf{R}) = \frac{h(\textbf{R} + \Delta\textbf{R}) - h(\textbf{R})}{h(\textbf{R})}.
	\setlength\belowdisplayskip{3pt}
\end{align}

Regarding~(\ref{Eq:xis}), we present the following lemma.

\begin{figure}[t]
	\centering
	\subfigure[RF IRS]{
		\centering
		\includegraphics[width=0.21\textwidth]{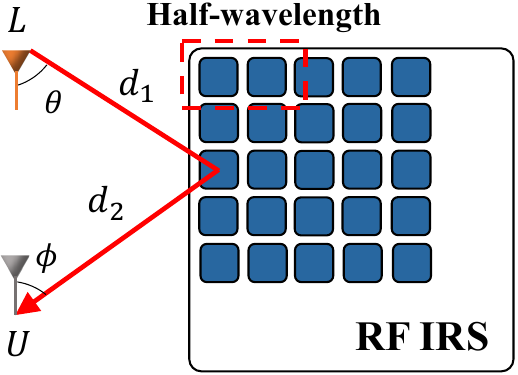}
		\setlength{\abovecaptionskip}{0.9cm}
		\label{Fig:correlation_RF}
	}
	\hspace{-0.5cm}
	\subfigure[OIRS]{
		\centering
		\includegraphics[width=0.24\textwidth]{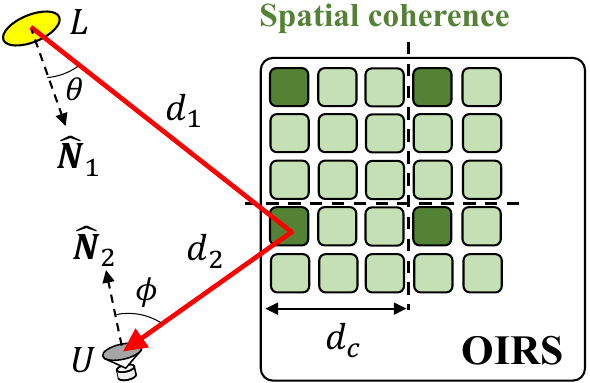}
		\setlength{\abovecaptionskip}{0.9cm}
		\label{Fig:correlation_OIRS}
	}
	\caption{The effect of space shift for the IRS-reflected channel gain: (a) Channel gain changes significantly due to the half-wavelength element spacing; (b) Channel gain shows spatial coherence due to the IM scheme adopted in VLC.}
	\label{Fig:correlation}
\end{figure}

\begin{lemma}
	\label{Lemma:growth rate}
	The growth rate of $\xi(\Delta\textbf{R})$ is given by
	\begin{align}
		\xi = &\Big\{ \frac{m}{d_1\cos\theta}\widehat{\boldsymbol{N}}_1 + \frac{1}{d_2\cos\phi}\widehat{\boldsymbol{N}}_2 - (\frac{m}{d_1} + \frac{2}{d_1 + d_2})\widehat{\textbf{LR}} - \notag\\
		&(\frac{1}{d_2} + \frac{2}{d_1 + d_2})\widehat{\textbf{UR}} \Big\}^T \Delta\textbf{R} 
		+ \Delta\textbf{R}^T \Big\{ \frac{-m}{d_1^2\cos^2\theta}\widehat{\boldsymbol{N}}_1\widehat{\boldsymbol{N}}_1^T - \notag\\
		& \frac{1}{d_2^2\cos^2\phi}\widehat{\boldsymbol{N}}_2\widehat{\boldsymbol{N}}_2^T + 2(\frac{m}{d_1^2} + \frac{1}{(d_1 + d_2)^2} + \frac{2}{d_1^2(d_1 + d_2)}) \notag\\
		& \times \widehat{\textbf{LR}}\widehat{\textbf{LR}}^T + 2(\frac{1}{d_2^2} + \frac{1}{(d_1 + d_2)^2}+\frac{2}{d_2^2(d_1 + d_2)})\widehat{\textbf{UR}}\widehat{\textbf{UR}}^T \notag\\
		& + \frac{2}{(d_1 + d_2)^2}(\widehat{\textbf{LR}}\widehat{\textbf{UR}}^T + \widehat{\textbf{UR}}\widehat{\textbf{LR}}^T) - (\frac{m}{d_1^2} + \frac{2}{d_1^2(d_1 + d_2)} \notag\\
		& + \frac{1}{d_2^2} + \frac{2}{d_2^2(d_1 + d_2)})\boldsymbol{I}_3 \Big\} \frac{\Delta\textbf{R}}{2},
	\end{align}
	where $\phi$ denotes the angle of incidence at the receiver with $\cos(\phi) = \widehat{\boldsymbol{N}}_2^T\widehat{\textbf{UR}}$.
\end{lemma}

\begin{proof}
	The proof is provided in Appendix~\ref{app:A}.
\end{proof}

Given a prescribed threshold $\xi_c > 0$, we can declare the OIRS-reflected channel as coherent if $|\xi(\Delta\textbf{R})| \leq \xi_c$. 
Therefore, for any given direction $\widehat{\Delta\textbf{R}_1}$, there exists a maximum coherence range from $\Delta\textbf{R}_1$ to its counterpart $\Delta\textbf{R}_2$ in a reverse direction.
As such, the OIRS coherence distance $d_c$ can be obtained as the minimum length over all directions, which is given by
\begin{align}
	d_c \triangleq \min\limits_{\begin{subarray}{c}
			\Delta\textbf{R}_1, \Delta\textbf{R}_2
	\end{subarray}} & \max\ \|\Delta\textbf{R}_1\|_2 + \|\Delta\textbf{R}_2\|_2 \label{Eq:d_c}\\
	\text{s.t.}\ &|\xi(\Delta\textbf{R}_1)| \leq \xi_c,\\
	& |\xi(\Delta\textbf{R}_2)| \leq \xi_c,\\
	& \Delta\textbf{R}_2 = -c \Delta\textbf{R}_1, \exists c>0,
\end{align}
and $d_c$ can be obtained based on the following lemma.

\begin{lemma}
	\label{Lemma:coherence distance}
	For any $A$, $B$, and $\xi_c>0$, let $r_1$ and $r_2$ denote two different roots of the equation $|\xi(r)| \triangleq |Ar^2 + Br| = \xi_c$ that are closest to zero.
	Let $d(A,B; \xi_c) \triangleq |r_2 - r_1|$, which is given by
	\begin{align}
		\label{Eq:r}
		d(A,B; \xi_c) = \left\{
		\begin{aligned}
			&\sqrt{\frac{B^2}{A^2}+\frac{4\xi_c}{\left| A \right|}},\quad \text{if}\ B^2 \leq 4|A|\xi_c,   \\
			&\frac{2\xi_c}{\left| B \right|},\quad \text{otherwise}.
		\end{aligned}
		\right.
	\end{align}
\end{lemma}

\begin{proof}
	The proof is provided in Appendix~\ref{app:B}.	
\end{proof}

Based on Lemma~\ref{Lemma:coherence distance}, the coherence distance of the OIRS can be obtained based on the following proposition.

\begin{proposition}
	\label{Prop:coherence}
	The OIRS coherence distance $d_c$ is given by 
	\begin{align}
		\label{Eq:d_c_expression}
		d_c = \min\limits_{\begin{subarray}{c}
				\widehat{\Delta\textbf{R}}
		\end{subarray}} d\left( \frac{\widehat{\Delta\textbf{R}}^T \boldsymbol{S}_2 \widehat{\Delta\textbf{R}}}{2}, \boldsymbol{s}_1^T \widehat{\Delta\textbf{R}}; \xi_c  \right),
	\end{align}
	where $\widehat{\Delta\textbf{R}}$ denotes the normalized space shift vector that is perpendicular to the OIRS.
\end{proposition}

\begin{proof}
	Given any direction $\widehat{\Delta\textbf{R}}$ such that $\Delta\textbf{R} = r\widehat{\Delta\textbf{R}}$, the growth rate can be re-expressed as $\xi(r)= r(\boldsymbol{s}_1^T \widehat{\Delta\textbf{R}}) + r^2(\widehat{\Delta\textbf{R}}^T \boldsymbol{S}_2\widehat{\Delta\textbf{R}})/2$ according to Lemma~\ref{Lemma:growth rate}.
	Next, to solve $|\xi(r)|\leq \xi_c$, the longest distance of $r$ can be obtained based on Lemma~\ref{Lemma:coherence distance} by setting $A = \widehat{\Delta\textbf{R}}^T \boldsymbol{S}_2 \widehat{\Delta\textbf{R}}/2$ and $B = \boldsymbol{s}_1^T \widehat{\Delta\textbf{R}}$.
	Therefore, the OIRS coherence distance in~(\ref{Eq:d_c}) can be derived by solving a series of subproblems that have different directions of $\widehat{\Delta\textbf{R}}$.
\end{proof}

\section{Proposed OIRS Channel Estimation Algorithm}
{\label{Sec:Proposed}}
Based on the spatial coherence property, in this subsection, we propose a spatial sampling-based algorithm to estimate the OIRS-reflected channel with low estimation overhead.
To this end, we first unify the OIRS association matrices $\boldsymbol{G}$ and $\boldsymbol{F}$ as a composite association matrix $\boldsymbol{V} = [\boldsymbol{v}_1, \cdots, \boldsymbol{v}_{N_tN_r}] \in\{0,1\}^{N\times N_tN_r}$, where $\boldsymbol{v}_{n_r + (n_t - 1)N_r} = \boldsymbol{f}_{n_r} \odot \boldsymbol{g}_{n_t}$.
As such, the OIRS-reflected channel can be re-expressed as~\cite{10190313}
\begin{equation}
	\setlength\abovedisplayskip{3pt}
	\label{Eq:MIMO_Vlinear}
	\text{vec}\left( \boldsymbol{H} \right) = \text{blkdiag}\left( \boldsymbol{V} \right)^T\text{vec}\left(\boldsymbol{H}_c\right),
	\setlength\belowdisplayskip{3pt}
\end{equation}
where $\text{blkdiag}(\boldsymbol{V})$ denotes a block diagonal matrix where the columns of $\boldsymbol{V}$ form the main diagonal elements in succession
Due to the fact that $\text{vec}(\boldsymbol{A}\boldsymbol{B}\boldsymbol{C}) = (\boldsymbol{C}^T \otimes \boldsymbol{A}) \text{vec}(\boldsymbol{B})$, the received signal in~(\ref{Eq:signal model_cohere}) can be vectorized as
\begin{align}
	\label{Eq:observed signal}
	\setlength\abovedisplayskip{3pt}
	\text{vec}\left(\boldsymbol{Y}\right) &= \left(\boldsymbol{X}^{T} \otimes \boldsymbol{I}_{N_r}\right)\text{vec}\left(\boldsymbol{H}\right) + \text{vec}\left(\boldsymbol{Z}\right) \notag\\
	&= \left(\boldsymbol{X}^{T} \otimes \boldsymbol{I}_{N_r}\right) \text{blkdiag}\left(\boldsymbol{V}\right)^T\text{vec}\left(\boldsymbol{H}_c\right) + \text{vec}\left(\boldsymbol{Z}\right).
	\setlength\belowdisplayskip{3pt}
\end{align}
Based on~(\ref{Eq:observed signal}), the entries of $\boldsymbol{H}_c$ can be estimated using some classical channel estimation algorithms by dynamically varying the OIRS association pattern $\boldsymbol{V}$.
However, note that $\rank(\boldsymbol{X}^{T} \otimes \boldsymbol{I}_{N_r}) = \rank(\boldsymbol{X}) \times \rank(\boldsymbol{I}_{N_r}) \leq N_t N_r$ and the total number of CSI parameters is $NN_tN_r$.
We need to perform at least $NN_tN_r / N_t N_r = N$ times of channel estimations, each corresponding to a subset of OIRS reflecting elements. 
To reduce the pilot overhead, we propose to leverage the OIRS coherence distance, as detailed below.

\addtolength{\topmargin}{0.05in}


\textbf{\textit{Phase \uppercase\expandafter{\romannumeral1:} Association pattern design.}}
As shown in Fig.~\ref{Fig:sample}, the OIRS reflecting elements are divided into $Q=Q_v \times Q_h$ subarrays according to the coherence distance $d_c$.
In each subarray, the reflecting element at its $n_r$th row and $n_t$th column is aligned with the $n_t$th LED and the $n_r$th PD, i.e., 
\begin{equation}
	\setlength\abovedisplayskip{3pt}
	f_{n_q, n_r} = g_{n_q, n_t} = 1,
	\setlength\belowdisplayskip{3pt}
\end{equation}
where $n_q$ denotes the index of the above reflecting element.
Note that there is no need to estimate the channel gains for other non-aligned OIRS reflecting elements, which will be derived by performing a spatial interpolation thanks to the spatial coherence.
Let $\Omega_{n_r,n_t}$ denote the index set of the aligned OIRS reflecting elements over all of its subarrays, which is given by
\begin{equation}
	\label{Eq:definition Omega}
	\setlength\abovedisplayskip{3pt}
	\Omega_{n_r,n_t} \triangleq \{n_q| \ q = 1,\cdots,Q \}.
	\setlength\belowdisplayskip{3pt}
\end{equation}
This OIRS association pattern is denoted by $\boldsymbol{G}^*$ and $\boldsymbol{F}^*$ and its equivalent parameter matrix is denoted by $\boldsymbol{V}^*$.

\begin{figure}[t]
	\centering
	\includegraphics[width=0.45\textwidth]{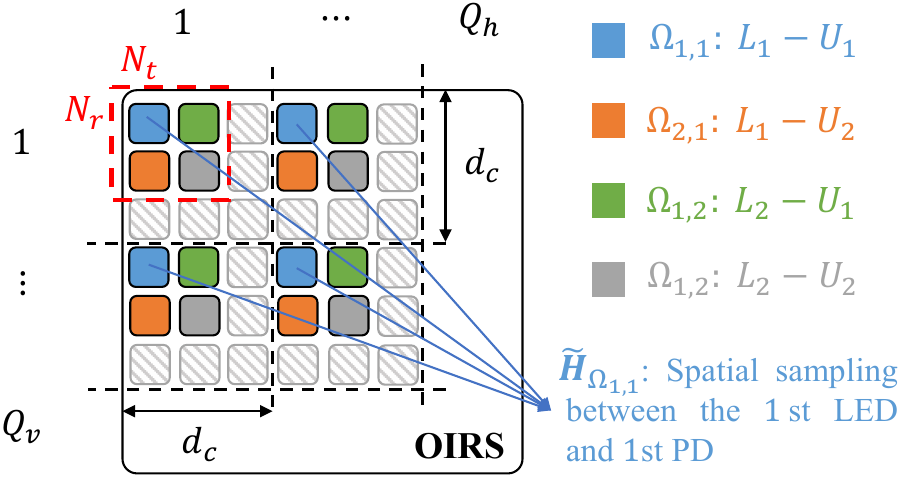}
	\caption{The proposed OIRS channel estimation based on spatial sampling.}
	\label{Fig:sample}
	\vspace{-0.2cm}
\end{figure}

\textbf{\textit{Phase \uppercase\expandafter{\romannumeral2:} Local channel estimation.}} 
In this phase, all OIRS subarrays are activated sequentially and the CSI parameters associated with each subarray are estimated.
Specifically, the $q$th subarray is activated in the $q$th time block, which is comprised of $P$ time slots, and we denote the composite association matrix in this time block as $\boldsymbol{V}_{q}^*$.
As the $n$th row of $\boldsymbol{V}_{q}^*$ represents the alignment of the $n$th OIRS element, we set the rows corresponding to the $q$th subarray the same as $\boldsymbol{V}^*$ while its other rows are set to zeros.
Let $\boldsymbol{H}_{q}$ denote the OIRS-reflected channel when the $q$th subarray is activated, which, based on~(\ref{Eq:MIMO_Vlinear}), is given by
\begin{equation}
	\label{Eq:relation}
	\setlength\abovedisplayskip{3pt}
	\text{vec}\left( \boldsymbol{H}_{q} \right) = \text{blkdiag}\left( \boldsymbol{V}_{q}^* \right)^T\text{vec}\left(\boldsymbol{H}_c\right).
	\setlength\belowdisplayskip{3pt}
\end{equation}
Similarly to~(\ref{Eq:signal model_cohere}), the received signal in the $q$th time block can be re-expressed as
\begin{equation}
	\label{Eq:phase 2}
	\setlength\abovedisplayskip{3pt}
	\boldsymbol{Y}_{q} = \boldsymbol{H}_{q} \boldsymbol{X}_{q} + \boldsymbol{Z}_{q},
	\setlength\belowdisplayskip{3pt}
\end{equation}
where $\boldsymbol{X}_{q}$ and $\boldsymbol{Z}_{q}$ denote the stacked pilot signals and AWGN over the $P$ time slots.
Based on~(\ref{Eq:phase 2}), $\boldsymbol{H}_{q}$ can be estimated by applying the MMSE estimator as
\begin{equation}
	\label{Eq:MMSE estimator}
	\setlength\abovedisplayskip{3pt}
	\boldsymbol{\tilde{H}}_{q} = \boldsymbol{Y}_{q}\boldsymbol{X}_{q}^T(\boldsymbol{X}_{q}\boldsymbol{X}_{q}^T + \sigma^2 \boldsymbol{I}_{N_t})^{-1},
	\setlength\belowdisplayskip{3pt}
\end{equation}
which includes $N_tN_r$ entries of $\boldsymbol{H}_c$ due to the low-rank property of $\boldsymbol{V}_{q}^*$.
Note that the entry of $\boldsymbol{\tilde{H}}_{q}$ at the $n_r$th row and the $n_t$th column, namely $[\boldsymbol{\tilde{H}}_{q}]_{n_r, n_t}$, is the estimated CSI between the $n_t$th LED and the $n_r$th PD via the corresponding OIRS reflecting element.

\begin{algorithm}[t]
	\caption{Proposed Spatial Sampling-based Algorithm}
	\label{Algo:sampling}
	\textbf{Input:} $N_t$, $N_r$, $N$, $d_c$\\
	\textbf{Output:} $\boldsymbol{G}^*$, $\boldsymbol{F}^*$, $\tilde{\boldsymbol{H}}_c$
	\begin{algorithmic}[1]
		\State Set $Q_v$, $Q_h$ based on $d_c$;
		\State Set subarray index $q \gets 1$;
		\Repeat
		\State Set $f_{n_q, n_r}^* \gets 1$ and $g_{n_q, n_t}^* \gets 1$, where $n_q$ denotes the index of the OIRS reflecting array at the $n_r$th row and $n_t$th column of the $q$th subarray;
		\State Update subarray index by $q \gets q + 1$;
		\Until{$q > Q$}
		\State Each OIRS reflecting element is configured to align with its associated LED and PD based on $\boldsymbol{G}^*$ and $\boldsymbol{F}^*$;
		\State Set $n \gets 1$, $n_r \gets 1$, and $q \gets 1$;
		\Repeat
		\State Activate the $q$th subarray;
		\State Local channel estimation to obtain $\boldsymbol{\tilde{H}}_{q}$ based on~(\ref{Eq:MMSE estimator});
		\State Update the index by $q \gets q + 1$;
		\Until{$q > Q$}
		\State Obtain partial CSI $\boldsymbol{\tilde{H}}_{\Omega_{n_r, n_t}}$ based on~(\ref{Eq:subarray cup});
		\Repeat
		\State Set LED index $n_t \gets 1$;
		\Repeat
		\State Interpolate full channel matrix $\boldsymbol{\tilde{H}}_{{n_r, n_t}}$ between 
		\Statex \qquad \quad the $n_r$th PD and the $n_t$th LED based on~(\ref{Eq:recover CSI});
		\State Recover CSI vector via 
		$\tilde{\boldsymbol{h}}_{n_r,n_t} = \text{vec}(\boldsymbol{\tilde{H}}_{{n_r, n_t}})$;
		\State Update LED index by $n_t \gets n_t + 1$;
		\Until{$n_t > N_t$}
		\State Update PD index by $n_r \gets n_r + 1$;
		\Until{$n_r > N_r$}
		\State Output the estimated channel matrix $\tilde{\boldsymbol{H}}_c = [\tilde{\boldsymbol{h}}_{1,1}, \cdots, \tilde{\boldsymbol{h}}_{N_r,N_t}]$;
	\end{algorithmic}
\end{algorithm}
\setlength{\textfloatsep}{0.5cm}

\textbf{\textit{Phase \uppercase\expandafter{\romannumeral3:} Full CSI recovery.}}
After completing all local channel estimations for the $Q$ subarrays, full OIRS-reflected CSI can be recovered by conducting a two-dimensional interpolation.
Without loss of generality, for the CSI between the $n_t$th LED and the $n_r$th PD, which are marked by the same color in Fig.~\ref{Fig:sample}, they can be expressed as the union of the corresponding entries of all $\boldsymbol{\tilde{H}}_{q}$, $q=1,\cdots, Q$, i.e.,
\begin{equation}
	\label{Eq:subarray cup}
	\setlength\abovedisplayskip{3pt}
	\boldsymbol{\tilde{H}}_{\Omega_{n_r, n_t}} = \mathop{\cup}\limits_{q}\left\{[\boldsymbol{\tilde{H}}_{q}]_{n_r, n_t}\right\}.
	\setlength\belowdisplayskip{3pt}
\end{equation}
These CSI is sufficient for full CSI recovery by leveraging the spatial coherence of OIRS.
Specifically, the channel gain of the reflecting elements at the $i$th row and the $j$th column can be retrieved as
\begin{equation}
	\label{Eq:recover CSI}
	\setlength\abovedisplayskip{3pt}
	[\boldsymbol{\tilde{H}}_{{n_r, n_t}}]_{i, j} = \sum_{q_v=1}^{Q_v}\sum_{q_h=1}^{Q_h} w(i, q_v; j, q_h) [\boldsymbol{\tilde{H}}_{\Omega_{n_r, n_t}}]_{q_v, q_h},
	\setlength\belowdisplayskip{3pt}
\end{equation}
where $w(\cdot)$ denotes the spline weight function.
The recovery of the CSI matrix $\boldsymbol{H}_c$ can be achieved by repeating the above process for different LED and PD antenna pairs by varying $n_t$ and $n_r$.

The proposed spatial sampling-based algorithm for OIRS channel estimation is summarized in Algorithm~\ref{Algo:sampling}.

\vspace{-0.1cm}
\section{Numerical Results}
\label{Sec:numerical}
\vspace{-0.1cm}
In this section, we provide numerical results to evaluate the efficacy of our proposed spatial sampling-based channel estimation method.
We consider the downlink VLC in an indoor environment with dimensions of 4m$\times$4m$\times$3m.
The transmitter is composed of multiple LEDs, which are in the shape of a circle with a radius of 0.1m and have a Lambertian index of 1.
The OIRS-reflected channel is generated by physical optics to emulate a practical scenario.
Specifically, an OIRS centered at (2m, 0m, 1.5m) is affixed to the wall along the $XoZ$ plane.
It is composed of $N = 24^2$ rectangular patches, each measuring 0.05m$\times$0.05m with a spacing of 0.1m. 
The reflectivity of each OIRS reflecting element is $0.9$.
Assuming the direct LoS signal is obstructed, the emitted light can only be received via the reflection from the OIRS mounted on the wall.
The receiver is composed of multiple PDs, where each PD is a square with the length of 0.05m and the spacing between PDs is $0.4$m.


\addtolength{\topmargin}{0.02in}

In Fig.~\ref{Fig:distance}, the spatial coherence of the OIRS is illustrated by showing the normalized channel gain versus the space shift.
In particular, a single LED and a PD are located at (1m, 2m, 3m) and (2m, 2m, 0m), respectively.
Meanwhile, an OIRS reflecting element is located at (1m, 2m, 1.5m) and its normal vector is fixed such that the cascaded LED-OIRS-PD path follows the reflection law.
For example, if the OIRS-reflected channel is defined as coherent when $\xi_c = 0.04$, it can be observed from the result that the OIRS coherence distance is nearly 0.4m.
It follows that we can employ this property to reduce the pilot overhead for OIRS channel estimation, as pursued in this paper.
Although the approximation model deviates from the real channel gain obtained through physical optics, the result shows that the coherence characteristic of OIRS exists.

\begin{figure}[t]
	\centering
	\includegraphics[width=0.38\textwidth]{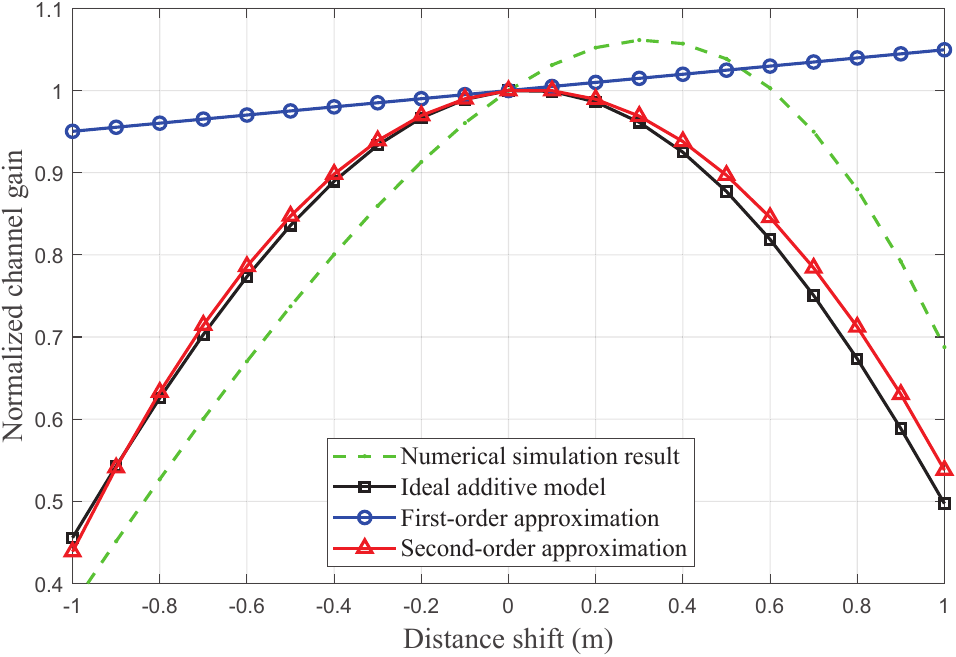}
	\caption{Normalized reflected channel gain versus the space shift over the OIRS.}
	\label{Fig:distance}
\end{figure}
\begin{figure}[t]
	\centering
	\includegraphics[width=0.42\textwidth]{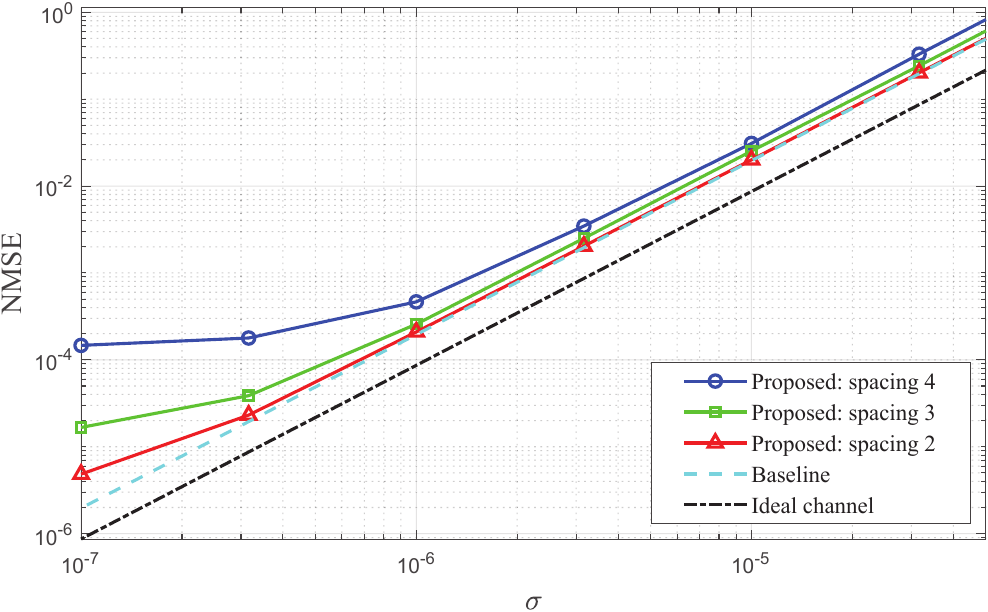}
	\caption{NMSE of the proposed OIRS channel estimation algorithm versus $\sigma$.}
	\label{Fig:NMSE_MIMO}
\end{figure}
\begin{figure}[t]
	\centering
	\includegraphics[width=0.42\textwidth]{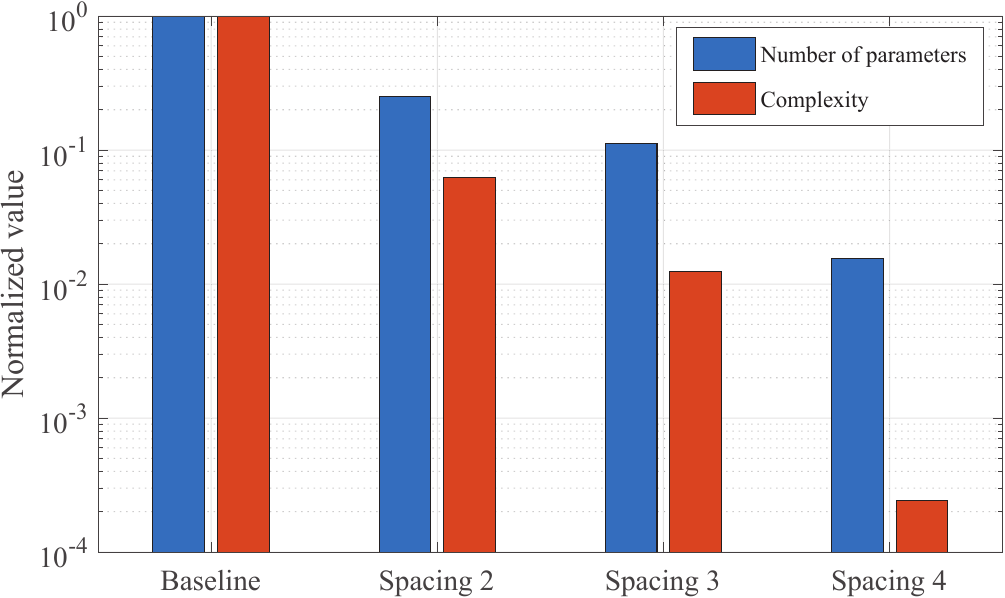}
	\caption{Overhead of the proposed channel estimation algorithm: Number of parameters and the computational complexity.}
	\label{Fig:overhead}
\end{figure}

In Fig.~\ref{Fig:NMSE_MIMO}, we compare the proposed channel estimation algorithm with other baseline schemes with $N_t = N_r = 2$.
In particular, we divide all reflecting elements into subarrays based on the coherence distance, i.e., there are 12$\times$12 sub-arrays with the spacing of 2.
Besides, there are 8$\times$8 sub-arrays and 6$\times$6 sub-arrays for the spacing of 3 and 4, respectively.
In each subarray, there are $N_t\times N_r = 4$ reflecting elements used for the local channel estimation.
Moreover, we send $P = 100$ pilot symbols in each time block.
On the contrary, there is no spatial sampling performed in the baseline scheme, thus resulting in the highest estimation accuracy. 
As shown in Fig.~\ref{Fig:NMSE_MIMO}, we plot the normalized mean square error (NMSE) versus the standard deviation of the noise $\sigma$ by different schemes.
It is observed that the NMSE of the proposed algorithm varies under different sampling spacing.
The NMSE gap between the proposed algorithm and the baseline scheme is negligible when the spacing is less than $d_c$, but becomes more significant when it is larger than $d_c$.
The above observation verifies the effectiveness of employing the spatial coherence property of the OIRS to simplify its channel estimation.

Finally, we show overhead of the proposed OIRS channel estimation method in Fig.~\ref{Fig:overhead} in terms of the complexity and CSI parameters.
It is observed that the number of CSI parameters decreases with the spacing, thanks to the use of the spatial sampling method.
In particular, when the sampling spacing is 2, the proposed algorithm estimates the OIRS-reflected channel gains every 4 OIRS reflecting elements, which is only a quarter of the number of CSI parameters to be estimated in the baseline scheme.
Next, the computational complexity of the MMSE estimator in~(\ref{Eq:MMSE estimator}) can be expressed as $\mathcal{O}(P(N_t^2 N_r + N_t^3 + N_t^2N_r)) = \mathcal{O}(PN_t^2(N_t + 2N_r))$.
It follows that significant complexity reduction can be achieved by our proposed channel estimation algorithm.
Even though, it is observed that the NMSE of the proposed channel estimation method increases with the sampling spacing, which implies a trade-off between the channel estimation accuracy and the pilot overhead.

\addtolength{\topmargin}{-0.05in}

\section{Conclusions}
{\label{Sec:Conclude}}
In this paper, we proposed an efficient channel estimation method for the OIRS-assisted VLC system.
Based on the OIRS association model, we first analyzed the effects of the space shift on the OIRS-reflected channel gain and derived the OIRS coherence distance. 
By leveraging this coherence, we proposed a spatial sampling-based algorithm to estimate the OIRS-reflected channel, where the subarray-level CSI is successively estimated to obtain the full CSI via interpolation.
Finally, numerical results validated our theoretical analysis and also showed the efficacy of our proposed channel estimation method in balancing the accuracy and overhead.

\begin{appendices}
	\section{Proof of Lemma 1}
	\label{app:A}
	According to~(\ref{Eq:lambertian}), the OIRS-reflected channel gain with space shift $\Delta\textbf{R}$ can be expressed as
	\begin{align}
		\label{Eq:h_spce_shift}
		h(\textbf{R} + \Delta\textbf{R}) \propto & \frac{1}{\left(\|\textbf{LR} + \Delta\textbf{R}\|_2 + \|\textbf{UR} + \Delta\textbf{R}\|_2\right)^2} \times \notag\\
		& \frac{\left(\widehat{\boldsymbol{N}}_1^T(\textbf{LR} + \Delta\textbf{R})\right)^m}{\|\textbf{LR} + \Delta\textbf{R}\|_2^m} 
		\frac{\widehat{\boldsymbol{N}}_2^T(\textbf{UR} + \Delta\textbf{R})}{\|\textbf{UR} + \Delta\textbf{R}\|_2}.
	\end{align}
	Based on~(\ref{Eq:xis}), the relative growth rate can be reformulated as $\xi(\Delta \textbf{R})=\boldsymbol{s}_1^T\Delta\textbf{R} + \Delta\textbf{R}^T\boldsymbol{S}_2\Delta\textbf{R}/2 + o\left(\|\Delta\textbf{R}\|^2\right)$ by the Taylor series, where the coefficient vector $\boldsymbol{s}_1$ is given by
	\begin{align}
		\label{Eq:spa-Taylor first}
		\boldsymbol{s}_1 = &\frac{m}{d_1\cos\theta}\widehat{\boldsymbol{N}}_1 + \frac{1}{d_2\cos\phi}\widehat{\boldsymbol{N}}_2 - \left(\frac{m}{d_1} + \frac{2}{d_1 + d_2}\right)\widehat{\textbf{LR}} \notag\\
		&- \left(\frac{1}{d_2} + \frac{2}{d_1 + d_2}\right)\widehat{\textbf{UR}},
	\end{align}
	and the coefficient matrix $\boldsymbol{S}_2$ is obtained as
	\begin{align}
		\label{Eq:spa-Taylor second}
		\boldsymbol{S}_2 =  &\frac{-m}{d_1^2\cos^2\theta}\widehat{\boldsymbol{N}}_1\widehat{\boldsymbol{N}}_1^T - \frac{1}{d_2^2\cos^2\phi}\widehat{\boldsymbol{N}}_2\widehat{\boldsymbol{N}}_2^T \notag\\
		&+ 2\left(\frac{m}{d_1^2} + \frac{1}{(d_1 + d_2)^2} + \frac{2}{d_1^2(d_1 + d_2)}\right)\widehat{\textbf{LR}}\widehat{\textbf{LR}}^T \notag\\
		&+ 2\left(\frac{1}{d_2^2} + \frac{1}{(d_1 + d_2)^2}+\frac{2}{d_2^2(d_1 + d_2)}\right)\widehat{\textbf{UR}}\widehat{\textbf{UR}}^T \notag\\
		&+ \frac{2}{(d_1 + d_2)^2}\left(\widehat{\textbf{LR}}\widehat{\textbf{UR}}^T + \widehat{\textbf{UR}}\widehat{\textbf{LR}}^T\right) \notag\\
		&-\left(\frac{m}{d_1^2} + \frac{1}{d_2^2} + \frac{2}{d_1^2(d_1 + d_2)} + \frac{2}{d_2^2(d_1 + d_2)}\right)\boldsymbol{I}_3.
	\end{align}
	
	\section{Proof of Lemma 2}
	\label{app:B}
	As shown in Fig.~\ref{Fig:inequality}, the extremum of function $|\xi(r)|$ can be obtained as $|\xi(-B/A/2)| = B^2/|A|/4$, and thus $|r_2 - r_1|$ can be derived in two cases.
	Firstly, when $B^2\leq 4|A|\xi_c$, there are two roots for the equation of $|\xi(r)| = \xi_c$ and $|r_2 - r_1|$ can be obtaiend by Weda's Theorem as
	\begin{equation}
		(r_2 - r_1)^2 = (r_2 + r_1)^2 - 4 r_1 r_2 =\frac{B^2}{A^2}+\frac{4\xi_c}{\left| A \right|},
	\end{equation}
	Secondly, when $B^2 > 4|A|\xi_c$, there are four roots and $|r_2 - r_1|$ can be upper-bounded by
	\begin{align}
		|r_2 - r_1| = 2 \times \frac{|r_2 - r_1|}{2} \geq \frac{2\xi_c}{\nabla \xi_t|_{\Delta t = 0}} = \frac{2\xi_c}{|B|},
	\end{align}
	where the inequality is due to the decreasing derivative of the function $\xi(r)$.
	
	\begin{figure}[t]
		\centering
		\includegraphics[width=0.45\textwidth]{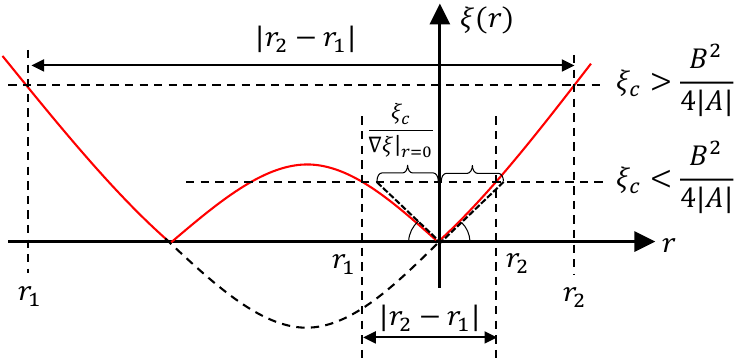}
		\caption{Sketch map of the proof of Lemma~\ref{Lemma:coherence distance}.}
		\label{Fig:inequality}
	\end{figure}
\end{appendices}

\balance
\small
\bibliographystyle{IEEEtran.bst}
\bibliography{IEEEabrv,ref}
\newpage
\vfill
\end{document}